\documentclass[11pt, a4paper]{article}

\usepackage[latin1]{inputenc}
\usepackage{a4}
\usepackage{latexsym}
\usepackage{amsfonts}
\usepackage{amsmath}
\usepackage{amsthm}
\usepackage{amssymb}
\usepackage{bbm}
\usepackage{verbatim}
\usepackage{algorithm}
\usepackage{algorithmic}
\usepackage{graphicx,color}  

\providecommand{\abs}[1]{\lvert#1\rvert}
\providecommand{\bigabs}[1]{\bigl\lvert#1\bigr\rvert}
\providecommand{\Bigabs}[1]{\Bigl\lvert#1\Bigr\rvert}

\providecommand{\norm}[1]{\lVert#1\rVert}

\providecommand{\ind}[1]{\mathbbm{1}_{#1}}

\newcommand{\com}[1]{}
\newcommand{\R}{\mathbb{R}}
\newcommand{\N}{\mathbb{N}}
\newcommand{\E}{\mathbb{E}}
\newcommand{\Prob}{\mathbb{P}}

\newcommand{\Normal}{\mathcal{N}}

\newcommand{\Ip}{\mathcal{J}^\Lambda_\Gamma}
\newcommand{\Ipk}{\mathcal{J}^{\Lambda_k}_{\Gamma_k}}

\newcommand{\dev}{\{\tau_i\leq t\}}
\newcommand{\indDev}{\ind{\dev}}
\newcommand{\pdev}{\Prob\bigl(\dev\bigr)}
\newcommand{\onot}{N_l^{[a,b]}(t)}
\newcommand{\trLoss}{F_l^{[a,b]}}
\newcommand{\Pik}{\Pi_{\Gamma_k}}

\newcommand{\ie}{{\it i.e.}\ }

\DeclareMathOperator{\card}{card}

\theoremstyle{plain}
\newtheorem{thm}{Theorem}
\newtheorem{prop}{Proposition}

\theoremstyle{definition}
\newtheorem{defn}{Definition}

\theoremstyle{remark}
\newtheorem*{remark}{Remark}

\author{{\sc Gilles Pag\`es} \thanks{Laboratoire de Probabilit\'es et Mod\`eles al\'eatoires, UMR~7599, Universit\'e Paris 6, case 188, 4,
pl. Jussieu, F-75252 Paris Cedex 5. E-mail: {\tt  gilles.pages@upmc.fr}} 
\quad {and} 
\quad {\sc Benedikt Wilbertz}\thanks{Laboratoire de Probabilit\'es et Mod\`eles al\'eatoires, UMR~7599, Universit\'e Paris 6, case 188, 4,
pl. Jussieu, F-75252 Paris Cedex 5. E-mail: {\tt
benedikt.wilbertz@upmc.fr}} }

\title{\bf Dual Quantization for random walks with application to credit
derivatives\com{the valuation of CDOs}\thanks{This work has been
supported by the CRIS project from the French pôle  de compétitivité ``Finance
Innovation''}}

\begin{document}

\maketitle

\bibliographystyle{plain}

\begin{abstract}
We propose a new Quantization algorithm for the approximation of inhomogeneous
random walks,
which are the key terms for the valuation of CDO-tranches in latent factor
models.
This approach is based on a dual quantization operator which posses an
intrinsic stationarity and therefore automatically leads to a second order
error bound for the weak approximation.
We illustrate the numerical performance of our methods in case of the
approximation of the conditional tranche function of synthetic CDO products and
draw comparisons to the approximations achieved by the saddlepoint method and
Stein's method.
\end{abstract}


\com{\bigskip
\noindent {\em MSC: }}

\section{Introduction}

In this paper we focus on the numerical approximation of inhomogeneous
Ber\-noul\-li random walks.

Therefore, let $(\Omega, \mathcal{F}, (\mathcal{F}_t), \Prob)$ be a filtered
probability space on which we define the {\em inhomogeneous random walk}
\begin{equation}\label{eq:defRW}
X := \sum_{i=1}^n \alpha_i Z_i,
\end{equation}
for some independent $\{0,1\}$-valued Bernoulli random variables $Z_i \sim
\mathcal{B}(p_i),\, p_i \in (0,1)$ and $\alpha_i > 0$.

The distribution of $X$ plays a crucial role for the valuation of basket credit
derivatives like CDO-tranches in latent factor models (see e.g. \cite{asb} or
\cite{gregory}). These are credit products, whose payoff is determined by
the loss in large portfolios of defaultable credit underlyings.

Therefore assume that we have a portfolio of $n$ defaultable credit names with
{\em notional amounts} $N_i$ and whose {\em default times} $\tau_i$ are
$(\mathcal{F}_t)$ stopping times, $i=1, \ldots, n$.
Here, $(\mathcal{F}_t)$ stands for the observable filtration of the credit
names. Moreover, we denote the {\em fractional recovery} of the $i$-th credit by $R_i$.

Hence, the {\em fractional loss} of the portfolio up to time $t$ is given by
\begin{equation}\label{eq:defLoss}
  l_t := \sum_{i=1}^n \frac{(1-R_i)N_i}{N} \indDev,
\end{equation}
where $N = \sum_{i=1}^n  N_j$ is the total notional.

Following the ideas of \cite{li} and \cite{frey}, the distributions of
the default events $\dev$ up to a fixed time $t$ are driven under the
risk-neutral probability measure by a common factor $U$ (which we may assume
w.l.o.g. as $\mathcal{U}([0,1])$ distributed) and some idiosyncratic noise
$\varepsilon_i$.

That means, that we assume that the events $\dev, \, i = 1,
\ldots, n$ are conditionally independent given $\sigma(U)$.

Furthermore, we require the existence of a copula function $F: [0,1]^2 \to
[0,1]$, such that $p \mapsto F(p,u)$ is a non-decreasing, right continuous
function for every $u \in [0,1]$ and 
\[
	\int_0^1 F(p,u) \, du = p, \quad p \in [0,1].
\] 

Since it holds
\begin{equation*}
\begin{split}
  \pdev = \E \bigl( \Prob(\dev | U ) \bigr) & = \int_0^1 \Prob\bigl(\dev | U = u
  \bigr)\, du
  \\
  	& = \int_0^1 F\Bigl(\pdev\, , u \Bigr)\, du,
\end{split}
\end{equation*}
we may interpret $F\bigl(\pdev\, , u \bigr)$ as the conditional
default probability $\Prob(\dev | U = u )$.

Typical choices for the function $F$ are the {\em standard Gaussian copula}
\[
	F(p,u) = \Phi \left\{ \frac{\Phi^{-1}(p) -
	\rho\Phi^{-1}(u)}{\sqrt{1-\rho^2}} \right\}
\]
with common correlation parameter $\rho$,
or the {\em Clayton copula} (cf. \cite{gregory}).

Thus, for a fixed time $t$, the risk-neutral conditional distributions of the
portfolio losses $l_t$ given the event $\{U=u\}$ are driven by a random walk of
type (\ref{eq:defRW}) with $\alpha_i := (1-R_i)N_i/N$ and conditionally independent
Bernoulli random variables $Z_i := \indDev$ with parameters $p_i :=
F\Bigl(\pdev\, , u \Bigr)$.

The cash flows of a (synthetic) CDO single tranche $[a,b]$ with attachment
points $0\leq a < b \leq 1$ read \com{can be simplified by assuming continuously
compounding} as follows: 

The protections seller of the tranche $[a,b]$ has to
pay at each default time $\tau_i$ which satisfies $l_{\tau_i}\in[a,b]$ the
notional of the defaulted name minus its recovery, i.e. 
\begin{align*}
	& & (1-R_i)N_i. && \text{({\em default leg})}
\end{align*}

On the other hand
he continuously receives a coupon payment of 
\begin{align*}
	& & \kappa \onot\,dt, && \text{({\em premium leg})}
\end{align*}
where $\kappa$ is the fair spread of the
tranche, which is to be determined by arbitrage arguments.
We denote by $\onot$ the {\it outstanding notional} of the tranche at time $t$, that is the notional
amount of the tranche $[a,b]$ which has not defaulted up to time $t$ .

Assuming a deterministic risk-free interest rate $r$ and continuously
compounding, we note that
\[
\frac{(1-R_i)N_i}{N} \ind{[a,b]}(l_{\tau_i}) = \trLoss(\tau_i) - \trLoss(\tau_i
-),
\]
where the tranche losses $\trLoss$ are defined as
\[
	\trLoss(t) := (l_t-a)^+ - (l_t-b)^+ = \begin{cases}
                                          0 &\text{if}\quad l_t < a\\
                                          l_t -a &\text{if}\quad  a\leq l_t
                                          \leq b\\ 
                                          b-a & \text{if}\quad l_t > b
                                          \end{cases}.
\] 
Hence, the discounted default payments accumulated
up to maturity $T$ maybe written as
 \begin{equation*}
 \begin{split}
 	\sum_{i=1}^n e^{-r\tau_i} (1-R_i)N_i \ind{[a,b]}(l_{\tau_i}) & = N
 	\sum_{i=1}^n e^{-r\tau_i} \bigl[ \trLoss(\tau_i) - \trLoss(\tau_i
	-) \bigr] \ind{\tau_i \leq T}\\ 
	&	= N \int_0^T e^{-rt} \trLoss(dt).
\end{split}
 \end{equation*}
 
 Concerning the premium leg, the outstanding notional $\onot$ of the tranche
 $[a,b]$ is given by
 \[
 \onot = N\cdot \bigl[(b-a) - \trLoss(t) \bigr] = \begin{cases}
                                          N\cdot(b-a) & \text{if}\quad l_t < a,
                                          \\ N\cdot(b-l_t) & \text{if}\quad
                                          a\leq l_t \leq b,\\
                                           0 & \text{if}\quad l_t > b
                                          \end{cases} 
 \]
so that the discounted coupon payments $\kappa\, e^{-rt} \onot\,dt$ accumulate
between $0$ and $T$ to
\[
\kappa \cdot N \int_{0}^T e^{-rt} \bigl[ (b-a) - \trLoss(t) \bigr] dt. 
\]

Under the risk-neutral probability measure both legs have to produce
an equal present value, i.e.
\[
N \int_0^T e^{-rt} \trLoss(dt) = \kappa \cdot N \int_{0}^T e^{-rt} \bigl[ (b-a)
- \trLoss(t) \bigr] dt,
\]
so that taking (risk-neutral) expectation and processing an integration by
parts yield the fair spread value $\kappa$, namely
\[
	\kappa = \frac{ e^{-rt} \E\trLoss(T) + r \int_0^T e^{-rt} \E \trLoss(t)\, dt
	}{\frac{b-a}{r} \bigl[1 - e^{-rT}] - \int_0^T e^{-rt} \E \trLoss(t)\, dt}.
\]
Here, the mathematical challenge consists in the computation of the
expectations $\E\trLoss(t)$.
This leads, within the latent factor models, 
to the approximation of the conditional expectations
\[
	\E(\trLoss(t)|U=u) = \E((l_t - a)^+|U=u) - \E((l_t - b)^+|U=u),
\]
since we have
\begin{equation}\label{eq:condLossInt}
	\E\trLoss(t) = \int_0^1 \E(\trLoss(t)|U=u)\, du.
\end{equation}
As already announced, the conditional distribution of $l_t$ is given by an
inhomogeneous random walk as defined in (\ref{eq:defRW}).

We therefore focus in this paper on the approximation of the distribution of
this type of random walks, the outer integral with respect to $U$ in
(\ref{eq:condLossInt}) can afterwards be approximated by standard quadrature
formulae.

For the usual applications $n$ has a size of about 100, which is by far too
large for an exact computation of the distribution of the random walk $X$, but
still too small to get accurate approximations based on the asymptotics
provided by limit theorems as $n$ goes to $\infty$.

Moreover, we have to deal in this general setting with arbitrary coefficients
$\alpha_i$, which destroy in general any recombining property of the random
walk. As a consequence, no (recombining) tree approach can be implemented.

So far, most approaches developed in the literature for the approximation of the
conditional tranche expectation $\E(\trLoss|U)$ rely upon the saddle point
method (cf. \cite{spThompson}) or an application of Stein's methods for both
Gaussian and Poisson approximation (cf. \cite{steinKaroui}).

\com{Unfortunately, both those methods do not provide any control of the induced
error and even do not allow to improve the accuracy of the approximation using
additional computational time.
Our new approach overcomes these drawbacks and moreover performs in a quite
robust way for even very inhomogeneous choices of the parameters $\alpha_i$ and
$p_i$.}

Although based on completely different mathematical tools, both approaches
suffer from the same lack of accuracy in the computation of
\[
	\E \Bigl( \sum_{i=1}^n \alpha_i Z_i - K \Bigr)^+
\]
when the {\it strike parameter} $K$ is ``at-the-mean'', i.e. when $\sum_{i=1}^n
\alpha_i p_i$ is close to $K$.
From a theoretical point of view no control of the induced error is \com{for
these methods }available. Finally, even if their numerical performances can be
considered as satisfactory in most situations, these approximations methods are
``static'': the ``design'' of the method cannot be modified to improve the
accuracy if a higher complexity is allowed.

The structure is as follows. In section 2 we introduce a new Dual Quantization
scheme for the approximation of the inhomogeneous random walk (\ref{eq:defRW}).
Moreover we establish error bounds for this approximation and discuss its
asymptotic behaviour. Section 3 is devoted to the numerical
implementation of this quantization scheme and its numerical performance.
Finally, in section 4, we give a slight modification of this
scheme to also capture the computation of sensitivities with respect to the
probabilities $p_i$ and the coefficients $\alpha_i$.

\section{Approximation of inhomogeneous Random Walks}
We will focus in this section on the numerical approximation of the
inhomogeneous random walk 
\[
X = \sum_{i=1}^n \alpha_i Z_i
\]
for independent  $Z_i \sim
\mathcal{B}(p_i),\, p_i \in (0,1)$ and $\alpha_i > 0$.

An exact computation of the distribution of $X$ is still not possible with
nowadays computers, since in our cases of interest we have $n \approx 100$ and
$X$ has up to $2^n$ states.\com{Since $X$ is a discrete random variable with up to $2^n$ states, so that in
the cases of interest ($n\approx 100$) this number is still much too large for
being processed by nowadays computers.} 
Hence we aim at constructing a random
variable $\widehat X$ with at most $N \ll 2^n$ states and which is close to $X$, e.g. $\E \abs{X-\widehat X}^2$ is small.

Due to the fact that there is no way to generate $X$ directly, we have to
construct approximations along the\com{ Markov chain
\[
	X^k := \sum_{i=1}^k \alpha_i Z_i
\]
so that
\[
	X^k = X^{k-1} + \alpha_k Z_k,
\]}
raondom walk
\begin{equation*}
\begin{split}
  X^0 & = 0,\\
  X^k & = X^{k-1} + \alpha_k Z_k, \quad k = 1, \ldots, n
\end{split}
\end{equation*}
where the increment $Z_k$ is an ordinary Bernoulli random variable which is easy
to handle.
Clearly we have
\[
	X = X^n
\]
and of course this would work similarly in full generality, if $X$ is a
function of a Markov chain.

Now suppose that we are equipped at each layer $k$ with some grid
$\Gamma_k = \{x_1^k, \ldots, x_{N_k}^k\}$ of size $N_k$ and a
(possibly random) projection operator $\Pik: \R \to \Gamma_k$, which maps the
r.v.'s $X^k$ into $\Gamma_k$.

We then may state a recursive approximation scheme for $X = X^n$ as follows
\begin{equation*}
\begin{split}
\widehat X^0 & := 0 \\
\widehat X^k & := \Pik(\widehat X^{k-\!1} + \alpha_k Z_k), \qquad k = 1, \ldots,
n.
\end{split}
\end{equation*}

This will be\com{Having outlined} the main principle for constructing
the approximation of $X^n$.\com{\widehat X = $\widehat X^n$, so that it} It
remains to choose appropriate grids $\Gamma_k$ and projection operators $\Pik$.
Here, it will turn out that the obvious choice of $\Pik$ as a nearest neighbor
projection is not sufficient in this setting and we will have to develop a new
approach.

\subsection{Quantization and Dual Quantization}
\paragraph{Regular Quantization}
In view of minimizing $\E \abs{X-\widehat X}^2$ for a general r.v.
$X\in L^2(\Prob)$, the above problem directly leads to the well known quadratic
quantization problem (cf. \cite{Foundations})
\begin{equation}\label{eq:DefQ}
\inf\Bigl\{ \E \abs{X - \widehat X}^2: \widehat X \text{ r.v. with
}\card\{\widehat X(\Omega)\} \leq N \Bigr\}
\end{equation}
at some level $N\in \N$. 
We will from now on call any discrete r.v. $\widehat X$ {\it Quantization} and
in particular if $\card\{\widehat X(\Omega)\} \leq N$ we call it {\it
$N$-Quantization}.
 
In fact one easily shows that (\ref{eq:DefQ}) is equivalent to solving
\begin{equation*}
\inf \Bigl\{ \E \min_{x \in \Gamma} \abs{X-x}^2 : \Gamma \subset \R,
\card\{\Gamma\} \leq N \Bigr\},
\end{equation*}
which means that $\Pi_\Gamma$ would be chosen as  a {\it nearest neighbor
projection} operator on $\Gamma$, \ie
\[
	\xi \mapsto \sum_{x \in \Gamma} x \cdot \ind{C_x(\Gamma)}(\xi),
\]
where $(C_x(\Gamma))_{x\in\Gamma}$ denotes a \com{\it
Voronoi-Partition}Borel-partition of $\R$ satisfying
\[
	C_x(\Gamma) \subset \bigl\{ \xi \in \R: \abs{\xi-x} \leq \min_{y \in \Gamma}
	\abs{\xi-y} \bigr\}.
\]
Such a partition is called {\it Voronoi-Partition} (of $\R$ related to
$\Gamma$).

In the one dimensional setting the {\em Voronoi cell} $C_{x_i}(\Gamma)$
generated by the ordered grid $\Gamma = \{x_1, \ldots, x_N\}$ consists simply of the interval
$[ \frac{x_{i-\!1} + x_i}{2},  \frac{x_{i} + x_{i+\!1}}{2} ]$. 
Nevertheless we
will use in this paper the more general notion of a Voronoi cell to
emphasize the underlying geometrical structure and the fact that this can also be defined in a
higher dimensional setting.

One shows (see \cite{Foundations}) that the infimum in (\ref{eq:DefQ}) actually
holds as an minimum: there exists an optimal quantization $\widehat
X^{\ast,N}$ (which takes exactly $N$ values if $X$ has infinite support).

Concerning the approximation of an expectation, first note that for $\Gamma :=
\widehat X(\Omega)$ we get
\begin{equation}\label{eq:cubature}
	\com{\E F(X) \approx }\E F(\widehat X) = \sum_{x\in\Gamma} F(x) \cdot
	\Prob(\widehat X = x),
\end{equation}
so that $\widehat X$ in fact induces a cubature formula with weights $\Prob(\widehat
	X = x),\, x \in \Gamma$.
This may provide a good approximation of $\E F(X)$, if $\widehat X$ is close to
the optimal solution of the quantization problem (\ref{eq:DefQ}).
	
\com{ Concerning the use of quantization as cubature formulae for the
approximation of expectations, we immediately derive for Lipschitz Functionals $F \in
C_\text{Lip}$}
For a Lipschitz functional $F \in
C_\text{Lip}(\R, \R)$ we immediately derive the error bound 
\[
\abs{\E F(X) - \E F(\widehat X)} \leq [F]_\text{Lip}\; \E \abs{X - \widehat X}.
\]

If moreover $F$ exhibits further smoothness properties, \ie $F \in C^1(\R)$
with Lipschitz derivative, we may establish for a quantization $\widehat X$
satisfying the stationarity property
\begin{equation}\label{eq:Stat}
\E(X|\widehat X) = \widehat X,
\end{equation}
\com{which is e.g. fulfilled by solutions to the optimal quantization problem
(\ref{eq:DefQ}),} a second order estimate (cf. \cite{pageSurv})
\[
	\abs{\E F(X) - \E F(\widehat X)} \leq [F']_\text{Lip}\; \E \abs{X - \widehat
	X}^2.
\]
Note that this stationarity property is always fulfilled if $\widehat X$ is
a solution to the optimal quantization problem (\ref{eq:DefQ}).

In view of the Zador Theorem (Thm 6.2 in \cite{Foundations}), which describes
the sharp asymptotics of the quantization problem (\ref{eq:DefQ}) as $N$ goes
to infinity, this leads to a quadratic error bound for an optimal quantization
$\widehat X^{\ast,N}$ of size $N$
\[
\abs{\E F(X) - \E F(\widehat X^{\ast, N})} \leq C_X\cdot [F']_\text{Lip} \cdot
N^{-2}.
\]

\com{As a matter of facts, }Unfortunately, in practice this stationarity
property (\ref{eq:Stat}) is only satisfied if $\widehat X$ is\com{ especially}
in some way optimized to ``fit'' the given distribution of $X$. This
optimization is time-consuming and due to the complicated structure of $X^n$ 
not feasible in our case of interest.

Hence we propose a (new) reverse interpolation operator to replace the nearest
neighbor projection, which offers an intrinsic stationarity and therefore leads
to a second order error bound without the need of adapting $\widehat X$ to the
exact distribution of $X$.

\subparagraph{Dual Quantization}
\com{Since }This alternative quantization approach for compactly supported
random variables is based on the Delaunay representation of a grid $\Gamma$,
which is the dual to its Voronoi diagram. Hence we will call this approach
{\it Dual Quantization}.

Suppose now to have an ordered grid $\Gamma$
\[
	a \leq x_1 \leq x_2 \leq \ldots \leq x_N \com{\leq x_{N+1}} \leq b
\]
for a r.v. $X$ with compact support included in $[a,b]$ (Typically, $[a,b]$ is
the convex hull of the support of $X$). Moreover we introduce for convenience
two auxiliary points $x_0 := a$ and $x_{N+1} := b$.

The {\it Delaunay tessellation} induced by $\Gamma$ then simply consists of the
line segments $\overline{x_j \, x_{j+1}},\, j = 0, \ldots, N+1$, where we
arbitrarily choose $\overline{x_j \, x_{j+1}}$ to be the half-open intervals
$[x_j, x_{j+1})$ for $j = 0, \ldots, N$ and $\overline{x_N \, x_{N+1}}$ as the
closed interval $[x_N , x_{N+1}]$. This way we arrive at a true partition of
the whole support of $X$.

To define a projection from $X(\Omega) \subseteq [a,b]$ onto $\Gamma$, we will not
just map any realization $X(\omega)$ to its nearest neighbor, but consider the two
endpoints of the line segment $\overline{x_{j^\ast}  x_{j^\ast\!+1}}$ into
which it falls.

We then perform a reverse random interpolation between these two points
$x_{j^\ast}, x_{j^\ast\!+1}$ in proportion to the ``barycentric coordinate''
\[
	\lambda := \frac{x_{j^\ast\!+1} - X(\omega)} {x_{j^\ast\!+1} -
	x_{j^\ast}},
\]
\ie we map $X(\omega)$ with probability $\lambda$  to $x_{j^\ast}$ and
with probability $(1-\lambda)$ to $x_{j^\ast\!+1}$ (see Figure
\ref{fig:revInt}). 
\begin{figure}[htbp]
    \begin{center}
      \resizebox{8.5cm}{!}{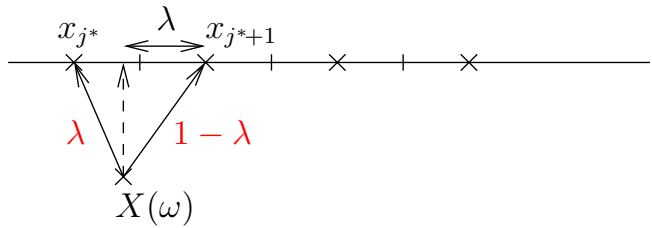}
    \end{center}
\caption{Reverse random Interpolation Operator $J^\Lambda$}
\label{fig:revInt}
  \end{figure}

A formal definition of this operator \com{can be }is given as follows.

\begin{defn}
Let $\Lambda \sim \mathcal{U}([0.1])$ be a r.v. on some probability
space $(\tilde \Omega, \tilde{\mathcal{F}}, \tilde \Prob)$
and let $\Gamma = (x_1,
\ldots, x_{N}),\, x_0 := a,\, x_{N+1}:=b$ be an ordered of $[a,b]$.
The {\it Dual Quantization operator} $\Ip$ \com{for the ordered grid $\Gamma =
(x_1, \ldots, x_{N})$ }is defined by
  \[
  \xi \mapsto \Ip(\xi) =\sum_{j=0}^{N}
  \Bigl(x_j\ind{\bigl[0,\frac{x_{j+1}-\xi}{x_{j+1}-x_j}\bigr)}(\Lambda) +
  x_{j+1}\ind{\bigl[\frac{x_{j+1}-\xi}{x_{j+1}-x_j},1\bigr]}(\Lambda)  \Bigr)
  \ind{\overline{x_j\, x_{j+1}}}(\xi).
  \]
\end{defn}

\begin{remark}
Note that we can always enlarge the original probability space
$(\Omega, \mathcal{F}, \Prob)$ to ensure that $\Lambda$ is defined on this
space and is independent of any r.v. defined on the
original space. Therefore we may assume w.l.o.g. that $\Lambda$ is defined on $(\Omega,
\mathcal{F}, \Prob)$.
\end{remark}

For $\xi \in [a,b]$ and $\overline{x_{j^\ast}  x_{j^\ast\!+1}}$
denoting the line segment into which $\xi$ falls, we get
\begin{equation}\label{eq:distIpk}
 \Prob(\Ip (\xi) = x_{j^\ast}) =
 \frac{x_{j^\ast\!+1}-\xi}{x_{j^\ast\!+1}-x_{j^\ast}} \quad\text{and}\quad
 \Prob(\Ip (\xi) = x_{j^\ast\!+1}) =
 1 - \frac{x_{j^\ast\!+1}-\xi}{x_{j^\ast\!+1}-x_{j^\ast}},
\end{equation}
so that $\Ip$ satisfies the desired reverse interpolation property.

As already announced, this Dual Quantization operator fulfills naturally a
stationarity property:

\begin{prop}[Stationarity]\label{prop:stat}
For any grid $\Gamma = (x_1, \ldots, x_{N})$ it holds
\[
\E(\Ip\!(X)|X) = X.
\]
\end{prop}
\begin{proof}
Let $\xi \in [a,b]$ and denote by $\overline{x_{j^\ast}  x_{j^\ast\!+1}}$ the
line segment in which $\xi$ falls.
\com{The assertion then follows immediately from}Then note that
\begin{equation*}
\begin{split}
\com{\E(\Ip\!(X)|X=x) & = \E \biggl(\sum_{j=1}^{N-1}
  \Bigl(x_i\ind{\bigl[0,\frac{x_{j+1}-x}{x_{j+1}-x_j}\bigr)}(\Lambda) +
  x_{j+1}\ind{\bigl[\frac{x_{j+1}-x}{x_{j+1}-x_j},1\bigr]}(\Lambda)  \Bigr)
  \ind{[x_j, x_{j+1}]}(X) | X= x  \biggr) \\
  & = x_{j^\ast} \Prob\biggl(\Lambda \in
  \Bigl[0,\frac{x^\ast_{j+1}-x}{x^\ast_{j+1}-x^\ast_j}\Bigr)\biggr) +
  x_{j^\ast\!+1} \Prob\biggl(\Lambda \in 
  \Bigl[\frac{x^\ast_{j+1}-x}{x^\ast_{j+1}-x^\ast_j},1\Bigr] \biggr) \\ & =
  \frac{1}{x^\ast_{j+1}-x^\ast_j} \Bigl(
  (x^\ast_{j+1}-x^\ast_j)(x-x^\ast_{j+1}) + x^\ast_{j+1}(x^\ast_{j+1}-x^\ast_j)
  \Bigr)\\
  & = x.}
  \E(\Ip\!(\xi)) & = \E \biggl(\sum_{j=0}^{N}
  \Bigl(x_i\ind{\bigl[0,\frac{x_{j+1}-x}{x_{j+1}-x_j}\bigr)}(\Lambda) +
  x_{j+1}\ind{\bigl[\frac{x_{j+1}-x}{x_{j+1}-x_j},1\bigr]}(\Lambda)  \Bigr)
  \ind{\overline{x_j, x_{j+1}}}(\xi) \biggr) \\
  & = x_{j^\ast} \cdot \Prob(\Ip (\xi) = x_{j^\ast}) +
  x_{j^\ast\!+1} \cdot \Prob(\Ip (\xi) = x_{j^\ast\!+1}) \\
  & =  \frac{1}{x_{j^\ast\!+1}-x_{j^\ast}} \Bigl(
  (x_{j^\ast\!+1}-x_{j^\ast})(\xi-x_{j^\ast\!+1}) +
  x_{j^\ast\!+1}(x_{j^\ast\!+1}-x_{j^\ast}) \Bigr)\\
  & = \xi.
\end{split}
\end{equation*}

The conclusion now follows from the independence of $X$ and $\Lambda$ which
implies
\[
\E(\Ip(X) | X) = \E(\Ip(\xi))_{|\xi = X} = X.
\]
\end{proof}

Similar to the primal Quantization setting we then derive by means of the
stationarity a second order estimate for the weak approximation of smoother
 integrands.

\begin{prop}\label{prop:secOrder}
  Let $F\in\mathcal{C}^1(\R)$ with Lipschitz derivative. Then every grid
  $\Gamma$ yields
\[
\abs{\E F(X) - \E F(\Ip\!(X))  } \leq [F']_\text{Lip}\; \E \abs{X - \Ip\!(X)
}^2.
\]
\end{prop}
\begin{proof}
From a Taylor expansion we derive
\[
	\bigabs{ F( \Ip(X)) - F(X) - F'(X)(\Ip(X) - X) } \leq [F']_\text{Lip}\;  \abs{X
	- \Ip\!(X) }^2
\]
so that the stationarity property (Proposition \ref{prop:stat}) implies
\[
	\bigabs{ \E \bigl( F( \Ip(X)) | X \bigr)  - F(X) } \leq
	[F']_\text{Lip}\; \E \bigl( \abs{X - \Ip\!(X) }^2 | X \bigr).
\]
Taking expectations then yields the assertion.
\end{proof}

\subsection{Application to the approximation of the inhomogeneous random walk}
\subsubsection{The algorithm}
We are now in the position to design an approximation scheme based on Dual
quantization in which the general projection operator $\Pik$ is replaced by the
dual quantization operator $\Ipk$. Let $\Gamma_1, \ldots, \Gamma_n$ be
some ordered grids. We set
\com{Coming back to the original approximation problem the
Dual Quantization scheme for the random walk $(X^k)_{1\leq k \leq n}$ now reads}
\begin{equation}\label{def:DQS}
   \begin{split}
     \widehat X^0 & := 0\\
     \widehat X^k & := \Ipk\!(\widehat X^{k-1} + \alpha_k Z_k), \qquad k = 1, \ldots,
n \\
  \end{split}
\end{equation}
\com{with }for $\Lambda_k\sim \mathcal{U}([0,1])$ i.i.d and independent of
$(Z_k)_{0\leq k \leq n}$.

We wish to approximate $\E F(X^n)$ by its dually quantized counterpart $\E
F(\widehat X^n)$.

\subsubsection{Error bound for the approximation of $\E F(\widehat X^n)$}

Concerning the approximation power of the dual quantization scheme
(\ref{def:DQS}) for $\E F(X^n)$ with $F\in \mathcal{C}^1_\text{Lip}(\R)$, 
we immediately derive
from Proposition \ref{prop:secOrder} the following local error bound for any
grid $\Gamma_k$ \com{quality of the expectation $\E F(X^n)$ for $F\in \mathcal{C}^1_\text{Lip}$\com{ integrands $F:\R\to \R$}, we immediately get
from Proposition \ref{prop:secOrder} for any grid $\Gamma_k$ the local error bound}
\[
\abs{ \E F(\widehat X^{k-1}+\alpha_k Z_k) - \E F(\widehat X^k) } \leq
[F']_\text{Lip}\;  \E \abs{ (\widehat X^{k-1}+\alpha_k Z_k ) - \widehat X^k 
}^2,
\]
since $\widehat X^k = \Ipk\!(\widehat X^{k-1} + \alpha_k Z_k).$

As a matter of fact, the global error then consists of all the local
insertion errors of the quantization operator along the random walk
$(X^k)_{1\leq k \leq n}$.\com{Then using the usual backward induction for the
BDP-principle leads to a global error estimate for this quantization scheme}

\begin{thm}[Global Error Bound]\label{thm:glbErr}
Let $F\in \mathcal{C}^1(\R)$ with Lipschitz derivative. \com{Then it holds for
any grid $\Gamma_k$ and the Dual Quantization scheme (\ref{def:DQS})}
Then the Dual Quantization scheme (\ref{def:DQS}) related to the grids
$\Gamma_k, 1\leq k \leq n$, satisfies
  \begin{equation*}
    \begin{split}
      \bigabs{\E F(X^{n}) - \E F(\widehat X^n) } & \leq [F']_\text{Lip}\;
      \sum_{k=1}^n \E \bigabs{(\widehat X^{k-1}+\alpha_k Z_k) - \widehat X^k 
      }^2\\
      & = [F']_\text{Lip}\; \sum_{k=1}^n \E \bigabs{ (\widehat
      X^{k-1}\!\!+\!\alpha_k Z_k) - \Ipk\!(\widehat
        X^{k-1}\!\!+\!\alpha_k Z_k)  }^2.
    \end{split}
  \end{equation*}
\end{thm}

\begin{proof}
First note that it follows from Proposition \ref{prop:secOrder} that for any
$\alpha \in \R$
\[
	\bigabs{\E F(X+\alpha) - \E F(\Ip(X) + \alpha)} \leq [F']_\text{Lip} \E
	\bigabs{X - \Ip(X)}^2.
\]
Consequently, we get for any r.v. $Z$ independent of $X$
\begin{equation*}
\begin{split}
\Bigabs{\E \Bigl[F(X  + & Z)   | Z=z\Bigr] - \E \Bigl[ F(\Ip(X) + Z)|Z=z \Bigr]}
\\ 
&  = \bigabs{\E F(X\!+\!z)  - \E F(\Ip(X) + z)} 
  \leq   [F']_\text{Lip} \E
	\bigabs{X - \Ip(X)}^2\\
\end{split}
\end{equation*}
and thus
\[
\bigabs{\E F(X+Z)  - \E F(\Ip(X) + Z)} 
  \leq   [F']_\text{Lip} \E
	\bigabs{X - \Ip(X)}^2.
\]
This finally yields
\begin{equation*}
\begin{split}
\bigabs{ \E F(\widehat X^n) - \E F(X^{n}) } & \leq \sum_{k=1}^n  \Bigabs{ \E
	F\bigl( \widehat X^k + \sum_{l=k+1}^n \!\alpha_l Z_l \bigr) - \E
	F\bigl( \widehat X^{k-1} + \sum_{l=k}^n \alpha_l Z_l \bigr)}\\
& = \sum_{k=1}^n  \Bigabs{ \E
	F\bigl( \Ipk(\widehat X^{k-1} \!+\! \alpha_k Z_k ) + \sum_{l=k+1}^n\! \alpha_l
	Z_l \bigr) \\
	& \qquad \qquad \quad - \,\E F\bigl( \widehat X^{k-1} \!+\! \alpha_k
	Z_k + \sum_{l=k+1}^n \!\alpha_l Z_l \bigr)}\\
& \leq [F']_\text{Lip}\; \sum_{k=1}^n \E \bigabs{ (\widehat
      X^{k-1}\!\!+\!\alpha_k Z_k) - \Ipk\!(\widehat
        X^{k-1}\!\!+\!\alpha_k Z_k)  }^2 \\
& = [F']_\text{Lip}\;
      \sum_{k=1}^n \E \bigabs{(\widehat X^{k-1}+\alpha_k Z_k) - \widehat X^k 
      }^2. 
\end{split}
\end{equation*}
\end{proof}

\subsection{Optimal choice of the grid $\Gamma$}\label{sec:grids}

\com{Concerning an optimal choice of the grids $\Gamma_k$ with respect to the
Dual Quantization operator $\Ip$, we arrive at the optimization problem
\begin{equation}\label{eq:DualOpt}
\E \abs{\Ip\!(X^k)-X^k  }^2 \to \min_{\Gamma \subset [a,b], \abs{\Gamma}\leq
N_k}
\end{equation}

In fact, the optimal quantization error achieved by the dual approach differs
from the original quantization for large $N$ only by a constant}

Let us temporarily come back to a static problem for an abstract random
variable $X$ with $\Prob(X\in [a,b]) = 1$.
In view of the second order estimate from Proposition \ref{prop:secOrder}, we
arrive for a fixed number $N\in\N$ at the optimization problem

\begin{equation}\label{eq:DualOpt}
\E \abs{X - \Ip\!(X)}^2 \to \inf_{\Gamma \subset [a,b], \abs{\Gamma}\leq
N}.
\end{equation}

It is established in \cite{dual} that this infimum actually stands as a
minimum. Hence optimal dual quantizers exists.
Moreover the mean dual quantization error achieved by such an optimal grid
differs from mean optimal quantization error of the primal quantization problem
(\ref{eq:DefQ}) asymptotically only by a constant.

\begin{thm}[Optimal rate (\cite{dual, Foundations})]\label{thm:optRate}
Let $X$ be a r.v. with $\Prob(X\in [a,b]) = 1$ and continuous density $\varphi$.
Then it holds
\[
	\lim_{N\to \infty} N^2 \inf_{\substack{\Gamma\subset [a,b]\\ \abs{\Gamma}\leq
	N}} \E \abs{X - \Ip(X)}^2 = 2 \lim_{N\to \infty}  N^2
	\inf_{\substack{\Gamma\subset [a,b]\\ \abs{\Gamma}\leq N}} \E
	\min_{x\in\Gamma} \abs{X-x}^2 = \frac{1}{6} 
	\biggl(\int_a^b \abs{\varphi(z)}^{3/2} dz
	\biggr)^{4/3}\com{\norm{\varphi}_{2/3}^2}.
\]

\end{thm}

\begin{remark}
This theorem about the asymptotics of the Dual Quantization problem can also be
generalized to non compactly supported r.v.'s. and to non quadratic mean error
(see \cite{dual}).
\end{remark}

\com{Unfortunately the generation of optimal grids for $X^k, k = 1, \ldots, n$
is out of reach so that we have to make a ``slightly'' sub-optimal decision:
We will choose $\Gamma_k $ to be grids which are optimal for a normal
distribution matching the first two moments of $X^k$,
since such a $\Normal(\mu_k, \sigma_k^2)$ distribution is for larger $k$ not
very far away from $X^k$.

Moreover our numerical observations seem to confirm the optimal $N^2$-rate for
these grids, which justifies our choice and suggest that it is more important
to construct a quantization scheme which is capable of generating the proper
weights for each of the grid points, than using the exact optimal grid within
this kind of discretization.}

Given the formula of the gradient and the hessian of the optimization problem
(\ref{eq:DualOpt}) with regard to $\Gamma$ a Newton algorithm similar to the one
described in \cite{pagesOQ} can be employed to construct numerically optimal
dual quantization grids.

\com{Concerning the numerical solution to the optimization problem
(\ref{eq:DualOpt}) a Newton algorithm as similar to the one described in
\cite{pagesOQ} can be employed.}

Nevertheless, a straightforward alternative is to derive an (only asymptotically
optimal) dual grid from a grid which is optimal for the primal quantization problem
(\ref{eq:DefQ}). Such grids are precomputed (cf. \cite{Website}) and
online available at
\begin{center}
{\tt www.quantization.math-fi.com}
\end{center}

To transform these regular quantization grids into dual ones,
we consider its midpoints, \ie if $y_1, \ldots, y_N$ denote an optimal grid for
the primal quantization problem (\ref{eq:DefQ}), we simply define its dual grid
\begin{equation}\label{eq:dualGrid}
	x_j := \frac{y_j + y_{j+1}}{2}, \quad j = 1, \ldots, N-1
\end{equation}
\com{for its dual grid.}
This choice is motivated by the asymptotic formula of Theorem
\ref{thm:optRate} and its proof in \cite{dual}, where exactly this midpoint rule
establishes a connection between dual and regular quantization.
Moreover, this connection allows to deduce the optimal rate for the dual
quantizers from that for regular quantizers.\\

Coming back to the problem of interest in this paper, the construction of
optimal (primal or dual) grids for each  $X^k, k = 1, \ldots, n$
is clearly out of reach so that we have to make a ``slightly'' sub-optimal
decision: we will choose grids which are optimal for a normal
distribution matching the first two moments of $X^k$,
since such a $\Normal(\mu_k, \sigma_k^2)$ distribution is close to $X^k$ for
large values of $k$. Additionally we can restrict these grids to the convex
hull of the support of $X^k$, \ie $[0, \sum_{i=1}^k \alpha_i]$.

Moreover, our numerical observations even tend to confirm an optimal
$N^{-2}$-rate for these sub-optimal grids. This emphazises again the importance of the
intrinsic stationarity provided by the dual quantization operator $\Ip$
in contrast to its primal counterpart, the nearest neighbor projection, where
the stationarity only holds for grids specially optimized for the true
underlying distribution, \ie the r.v. $X^k$ in our case.

\section{Numerical implementation and results}

\subsection{Numerical Implementation}

We now present numerical results and notes on the implementation of the Dual
Quantization scheme (\ref{def:DQS}) for the approximation of 
\begin{equation}\label{eq:Expectation}
	\E \Bigl( \sum_{i=1}^n \alpha_i Z_i - K \Bigr)_+,
\end{equation}
by means of
\begin{equation}\label{eq:approx}
	\E ( \widehat X^n - K)_+.
\end{equation}

Concerning the second order estimate of Theorem \ref{thm:glbErr}, the call
function $x \mapsto x_+$ clearly does not satisfy the assumptions of a
continuously differentiable function with Lipschitz derivative.
Nevertheless\com{it is always possible to approximate $(x)^+$ by a set of
functions which fulfill the smoothness properties of Theorem \ref{thm:glbErr}, which
explains why obtain in our numerical experiments still the optimal second order
rate of $N^{-2}$.}
we can replace $x_+$ by $\varphi_\varepsilon(x) := \E(x+\varepsilon Y)_+$, where
$Y\sim \Normal(0,1)$ and $\varepsilon > 0$ to overcome this shortcoming. This
function satisfies $\abs{\varphi_\varepsilon(x) - x_+}\leq \varepsilon,\,
\varphi_\varepsilon \in \mathcal{C}^\infty(\R)$ and $0 \leq \varphi'_\varepsilon \leq 1$. 
Furthermore $\varphi_\varepsilon$ writes $\varphi_\varepsilon(x) = x\cdot
\Phi(\frac{x}{\varepsilon}) + \frac{\varepsilon}{\sqrt{2\pi}} e^{-\frac{x^2}{2
\varepsilon^2}}$, where $\Phi$ is the distribution function of the standard
normal distribution.
\\

We could imagine to compute $\E(\widehat X^n - K)_+$ using a {\em backward}
dynamic programming formula based on (\ref{def:DQS}). However such an approach is
``payoff'' dependent and consequently time-consuming since the computation
needs to be done for many values of $K$ as emphasized in the introduction.

An alternative is to directly rely on the cubature formula
\[
\E (\widehat X^n - K)_+ = \sum_{i = 0}^{N_n+1} (x^n_i - K)_+ \cdot
	\Prob(\widehat X^n = x^n_i	)
\]
to approximate (\ref{eq:Expectation}). Here $\Gamma_n = \{ x_1^n,\ldots,
x_{N_n}^n \}$ is a dual grid of a normal distribution as described by
(\ref{eq:dualGrid})  in section \ref{sec:grids} and we set $x_0^n := 0,\,
x^n_{N_n+1} := \sum_{i=1}^n \alpha_i$.

The main task is then to compute the weights $\Prob(\widehat X^n = x_j^n)$ for
$1 \leq j \leq N_n$, which are given by the following {\em forward} recursive
formula.
\begin{prop}\label{prop:weights} 
In the dual quantization scheme (\ref{def:DQS}) the weights $\Prob(\widehat X^k
= x_l^k ) $ satisfy\com{
 For the weights $\Prob(\widehat X^k = x_l^k ) $ it holds within the
dual quantization scheme (\ref{def:DQS})}
\begin{equation*}
\begin{split}
	\Prob(\widehat X^k = x_l^k ) = \sum_{j=0}^{N_k+1} \Bigl[ ( 1 & - p_k)\cdot
	\lambda_l^k(x_j^{k-\!1}) \cdot \Prob(\widehat X^{k-\!1} = x^{k-\!1}_j) \\
	  & + p_k \cdot
	\lambda_l^k(x_j^{k-\!1}\! + \alpha_k) \cdot \Prob(\widehat X^{k-\!1} =
	x^{k-\!1}_j) \Bigr],
	\end{split}
\end{equation*}
where
\[
	\lambda_l^k(\xi) = \Prob\bigl(\Ipk(\xi) = x_l^k\bigr) = \begin{cases}                                                                                                                                          
                                               \frac{x_{j^\ast\!+1}-\xi}{x_{j^\ast\!+1}-x_{j^\ast}}, & \text{if } x_l^k = x_{j^\ast}^k\\
                                               1 - \lambda^k_{l-\!1}(\xi), &  \text{if } x_l^k = x_{j^\ast+\!1}^k\\
                                               0 & \text{otherwise} 
                                               \end{cases},
\]
and $j^\ast := j^\ast(\xi)$ denoting the line segment, which satisfies $\xi
\in \overline{x^k_{j^\ast}\, x^k_{j^\ast+1}}$.
\end{prop}
\begin{proof}
 We clearly have 
 \begin{equation*}
 \begin{split}
 \Prob(\widehat X^k = x_l^k ) & = \sum_{j=0}^{N_k+1} \Prob(\widehat X^k =
 x_l^k | \widehat X^{k-\!1} =	x^{k-\!1}_j ) \cdot \Prob(\widehat X^{k-\!1} =
 x^{k-\!1}_j)\\ 
 & = \sum_{j=0}^{N_k+1} \Prob(\Ipk(X^{k-1} + \alpha_k Z_k ) =
 x_l^k | \widehat X^{k-\!1} =	x^{k-\!1}_j ) \cdot \Prob(\widehat X^{k-\!1} =
 x^{k-\!1}_j).
 \end{split}
 \end{equation*}
 
 Since $\Lambda_k, Z_k$ are independent of $\widehat X^{k-1}$, we derive
 \begin{equation*}
 \begin{split}
 \Prob(\Ipk(X^{k-1} + \alpha_k Z_k ) = x_l^k | \widehat X^{k-\!1} =	x^{k-\!1}_j ) 
 & = \Prob(\Ipk(x^{k-\!1}_j + \alpha_k Z_k ) = x_l^k  )\\ 
 & = (1 - p_k) \cdot \Prob(\Ipk(x^{k-\!1}_j ) = x_l^k  )\\
 & \quad\quad +  p_k \cdot \Prob(\Ipk(x^{k-\!1}_j + \alpha_k ) = x_l^k  )
 \end{split}
 \end{equation*}
so that finally (\ref{eq:distIpk}) yields the assertion.
\end{proof}

\subparagraph{Practical implementation}
From an implementational point of view we process the quantization scheme
(\ref{def:DQS}) and \com{
This can be done by processing the quantization scheme
(\ref{def:DQS}) in a forward way.  

\com{

From  (\ref{eq:cubature}) it follows that we approximate
(\ref{eq:Expectation}) by the cubature formula
\[
	\E ( \widehat X^n - K)^+ = \sum_{j=1}^{N_n} ( x_j^n - K)^+ \cdot
	\Prob(\widehat X^n = x_j^n),
\] 
where we choose $\Gamma_n = \{x_1^n, \ldots, x_{N_n}^n\}$ as the dual grid of a
normal distribution as described in section \ref{sec:grids}.

Hence it remains to compute the weights $P(\widehat X^n = x_j^n)$ for
$1 \leq j \leq N_n$. This is done while processing the quantization scheme
(\ref{def:DQS}). 
}

Therefore }we start with
\[
	\widehat X^0 = 0,
\]http://mathema.tician.de/dl/pub/pycuda-mit.pdf
i.e. a grid $\Gamma_0 = \{ x_1^0\} = \{ 0 \}$ and weight 
\[
	w_1^0 :=\Prob(\widehat X^0 = x_1^0) = 1.
\]

To pass from time $k-1$ to $k$, we suppose to have grids $\Gamma_k = \{
x_1^k,\ldots, x_{N_k}^k \}$ as described above\com{as dual grid of a normal
distribution matching the first two moments of $X^k$}. Additionally, we add the
endpoints $x_0^k := 0$ and $x^k_{N_k+1} := \sum_{i=1}^k \alpha_i$ and define $\overline \Gamma_k :=
\Gamma_{k} \cup \{x_0^{k},x_{N_{k}+1}^{k}\}$. 
Moreover we assume that the weights 
\[
	w_j^{k-1} = \Prob(\widehat	 X^{k-1} = x_j^{k-1}) , \quad j = 0, \ldots,
	N_{k-1}+1
\]
have already been computed.\\

We then could compute $\Prob(\widehat	 X^{k} = x_l^{k}), \, l = 0,\ldots,
N_k+1$ directly by means of Proposition \ref{prop:weights}. 
However, this approach requires $2(N_k+2)(N_{k-1}+2)$ evaluations of the
barycentric coordinate $\lambda^k_l$, 
which are not cheap operations, since each evaluation involves a nearest
neighbor search to find the matching line segment $\overline{x^k_{j^\ast}\,
x^k_{j^\ast+1}}$.

Therefore it is more efficient to first iterate through the state space
\[
	\overline \Gamma_{k-1} \cup \bigl(\overline \Gamma_{k-1} + \alpha_k\bigr)
\] 
of the r.v. $\widehat X^{k-1} + \alpha_k Z_k$.
While computing for each $\xi \in \overline \Gamma_{k-1} \cup \overline
\Gamma_{k-1} + \alpha_k$ its matching line segment $\overline{x^k_{j^\ast}\,
x^k_{j^\ast+1}}$ in the grid $\Gamma_k$, we directly update the weight
vector $(w_l^k)_{0 \leq l \leq N_k+1}$ at positions $l=j^\ast(\xi)$ and
$l=j^\ast(\xi) + 1$.

This approach is given by Algorithm \ref{alg} and needs only $2
(N_{k-1}+2)$ nearest neighbor searches per layer $k$.

\begin{figure}[htbp]
    \begin{center}
      \resizebox{8.5cm}{!}{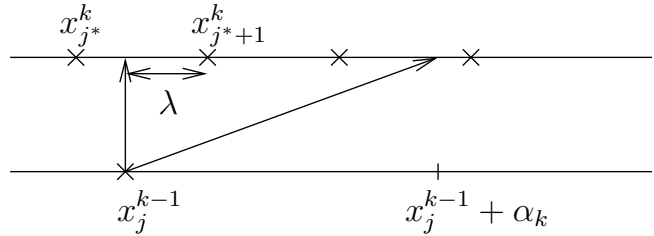}
    \end{center}
\caption{Weight-updating}
\label{fig:weightUp}
  \end{figure}

\com{
We then rely on the
induction formula
\com{We then proceed by applying the step}
\[
	\widehat X^k = \Ipk\!(\widehat X^{k-1} + \alpha_k Z_k),
\]
where $\widehat X^{k-1} + \alpha_k Z_k$ is a r.v. with state space
\com{[
	\Gamma_{k-1} \cup \{x_0^{k-1},x_{N_{k-1}+1}^{k-1}\} 
		\cup (\Gamma_{k-1} + \alpha_k) 
			\cup \{x_0^{k-1}+ \alpha_k,x_{N_{k-1}+1}^{k-1}+ \alpha_k\},
\]}
$
	\overline \Gamma_{k-1} \cup (\overline \Gamma_{ method for
	k-1} + \alpha_k),
$ 
since $Z_k$ is $\{0,1\}$-valued.

If we assume that the weights
\[
	w_j^{k-1} = \Prob(\widehat	 X^{k-1} = x_j^{k-1}) , \quad j = 0, \ldots,
	N_{k-1}+1
\]
have already been computed,
we then apply the random operator $\Ipk$ on the points $x^{k-1}_j $ and $
x^{k-1}_j + \alpha_k$ with associated probabilities $(1-p_k) \cdot w_j^{k-1}$
and $p_k \cdot w_j^{k-1}$ for $1\leq j\leq N_{k-1}$.

\begin{figure}[htbp]
    \begin{center}
      \resizebox{8.5cm}{!}{\input weightsUpInt.pdf_t}
    \end{center}
\caption{Weight-Updating}
\label{fig:weightUp}
  \end{figure}

That means we perform for the case $Z_k = 0 $ the weight updating at $x:=
x_j^{k-1}$ as
\begin{equation*}
\begin{split}
	 w^k_{j^\ast}  & = \lambda \cdot (1 - p_k) \cdot w_j\\
	 w^k_{j^\ast+1}  & = (1 - \lambda) \cdot (1 - p_k) \cdot w_j,
\end{split}
\end{equation*} 
where 
\[
	\lambda = \frac{x^k_{j^\ast\!+1} - x} {x^k_{j^\ast\!+1} -
	x^k_{j^\ast}}
\]
and $\overline{x_{j^\ast}  x_{j^\ast\!+1}}$ denotes the line segment in which
$x$ falls (this can be determined by a divide and conquer nearest neighbor
search).

In the case $Z_k = 1$, i.e. $x = x_j^{k-1} + \alpha_k$ this procedure
consequently reads
\begin{equation*}
\begin{split}
	 w^k_{j^\ast}  & = \lambda \cdot p_k \cdot w_j\\
	 w^k_{j^\ast+1}  & = (1 - \lambda) \cdot p_k \cdot w_j.
\end{split}
\end{equation*} 

By accumulating these weights over all possible states $x^{k-1}_j,
x^{k-1}_j + \alpha_k, 1\leq j\leq N_{k-1}$ we arrive at the discrete
distribution $(w_j^k, x_j^k)_{1\leq j\leq N_k}$ of $\widehat X^k$.

Applying the above step successively for all layers $k = 1, \ldots, n$ we
finally establish the whole distribution of $\widehat X^n$ and may compute the
approximation of $\E F(X^n )$ as
\[
	\sum_{j=1}^{N_n} w_j^n \cdot F(x^n_j).
\]

The complete weight computation reads as follows in an algorithmic description:
}
\begin{algorithm}
\caption{Weight-Computation for the Dual Quantization scheme
(\ref{def:DQS})}\label{alg}
\begin{algorithmic}
\STATE {\# {\bf Initialization}}
\STATE $\Gamma_0  \leftarrow \{ 0 \}$
\STATE $w^0_1  \leftarrow 1$
\STATE
\FOR{$k=1, \ldots, n$}
	\STATE 
	\STATE $w_j^{k} \leftarrow 0, \quad j = 0, \ldots, N_k+1$
	\STATE
	\FOR{$j = 0, \ldots, N_k+1$}
		\STATE {\# {\bf Case: $Z_k = 0$}}
		\STATE Find line segment $\overline{x^k_{j^\ast}  x^k_{j^\ast\!+1}}$ in which 
		$x_j^{k-1}$ falls
		\STATE $\lambda \leftarrow \frac{x^k_{j^\ast\!+1} -
		x_j^{k-1}}{x^k_{j^\ast\!+1} - x^k_{j^\ast}}$ 
		\STATE Set		
		\STATE $\quad w^k_{j^\ast}  +\!\!= \lambda \cdot (1 - p_k) \cdot w_j^{k-1}$
	 	\STATE $\quad w^k_{j^\ast+1}  +\!\!= (1 - \lambda) \cdot (1 - p_k) \cdot	w_j^{k-1}$
		\STATE
		\STATE {\# {\bf Case: $Z_k = 1$}}
		\STATE Find line segment $\overline{x^k_{j^\ast}  x^k_{j^\ast\!+1}}$ in which 
		$x_j^{k-1}+\alpha_k$ falls
		\STATE $\lambda \leftarrow \frac{x^k_{j^\ast\!+1} -
		(x_j^{k-1}+\alpha_k)}{x^k_{j^\ast\!+1} - x^k_{j^\ast}}$
		\STATE Set		
		\STATE $\quad w^k_{j^\ast}  +\!\!= \lambda \cdot p_k \cdot w_j^{k-1}$
	 	\STATE $\quad w^k_{j^\ast+1}  +\!\!= (1 - \lambda) \cdot p_k \cdot w_j^{k-1}$		
	\ENDFOR
	\STATE
\ENDFOR
\end{algorithmic}
\end{algorithm}

\subsection{Speeding up the procedure}

\subsubsection{Aggregation of insertion steps}

In view of the global error bound
from Theorem \ref{thm:glbErr} it is useful to reduce the number of grid
insertion steps $\Ipk$. A natural way to do so, is to aggregate $n_0$ r.v. $Z_i$
into
      \[
      Z'_k = \sum_{i=(k-1)n_0+1}^{k\cdot n_0} \alpha_i Z_i, \qquad k =
      1, \ldots, n/n_0
      \]
      and then set
      \[
      \widehat X^k = \Ipk\!(\widehat X^{k-1} + Z'_k), \quad k = 1, \ldots,
      n/n_0.
      \]

      E.g. with a choice of $n_0 = 2$ we would insert a binomial r.v.
      with $4$ states at every grid-point, but performing only
      $1/2$ of the insertions.
 
 However, for this choice of $n_0 = 2$ the overall number of nearest neighbor
 searches for the matching line segment $\overline{x_{j^\ast}  x_{j^\ast\!+1}}$
 remains the same as for $n_0 = 1$.
 
\subsubsection{Romberg extrapolation}

An additional improvement of this method is based on the heuristic guess that
the approximation of $\E F(X)$ by $\E F(\widehat X^n)$ (see Proposition
\ref{prop:secOrder}) admits a higher order expansion
\begin{equation}\label{eq:assumpRom}
      \E F(X) = \E F(\widehat X^n) + \kappa N^{-2} + o(N^{-2}).
\end{equation}      
      
We then may use quantization grids of two different sizes $N_1 \ll N_2$ to
cancel the second order term $\kappa N^{-2}$ in the above representation.

This leads to the {\it Romberg extrapolation} formula
\begin{equation}\label{eq:romberg}
 \E F(X) = \frac{N_1^2\, \E F(\widehat X^n_{N_1}) - N_2^2\, \E F(\widehat
 X^n_{N_2}) }{N_1^2 - N_2^2} + o(N_2^{-2}).
\end{equation}

Although assumption (\ref{eq:assumpRom}) is only of a heuristic nature,
numerical results seem to confirm this conjecture (like for ``regular''optimal
quantization).
 
 \subsection{Numerical experiments}
 
 For the numerical results we implemented the above dual quantization scheme for
 grids of constant size $500$ and $1000$ in all layers $k = 1, \ldots, n$.
  Regarding the Romberg extrapolation approach we applied the extrapolation
  formula (\ref{eq:romberg}) for sizes $100$ and $500$. 

As concerns methods to compare our approach to, we implemented a
saddlepoint-point method (cf. \cite{spThompson}) and the Stein approach for a Poisson and Normal
approximation developed in \cite{steinKaroui}. 
 
 We tested two typical situations: homogeneous and truly inhomogeneous
 Bernoulli random walks.
 
 \subparagraph{Homogeneous random walk.}
 
 Let us start with a homogeneous Test-Scenario, i.e. all $\alpha_i$ are chosen
 equal to $1$. Moreover we assume the $p_i$ to be a $n$ sample of a log-normal
 distribution, which corresponds to the case of a Gaussian copula.
 Hence the parameters read as follows:
 
  \begin{itemize}
  \item $n= 100$,
  \item $\alpha_i = 1$,
  \item $p_i = p_0 \, \exp(\sigma \xi_i - \sigma^2/2), \quad \xi_i \sim
  \Normal(0,1)$ i.i.d.,

    with $p_0 \in \{0.05, 0.1, 0.2 \}, \; \sigma = 0.5$,
  \end{itemize}

 Since this setting yields a recombining binomial tree, we can compute the
 exact reference values of
 \[
\E \Bigl( \sum_{i=1}^n Z_i - K\Bigr)^+
\]
for $K\in[0,50]$ and plot the absolute errors as a function of the strike $K$
to illustrate the numerical performances of the methods. This has been reported
in Figures \ref{fig:homo1} to \ref{fig:homo3}.

\begin{figure}
    \centering
    \includegraphics[width=0.62\textwidth,angle=270]{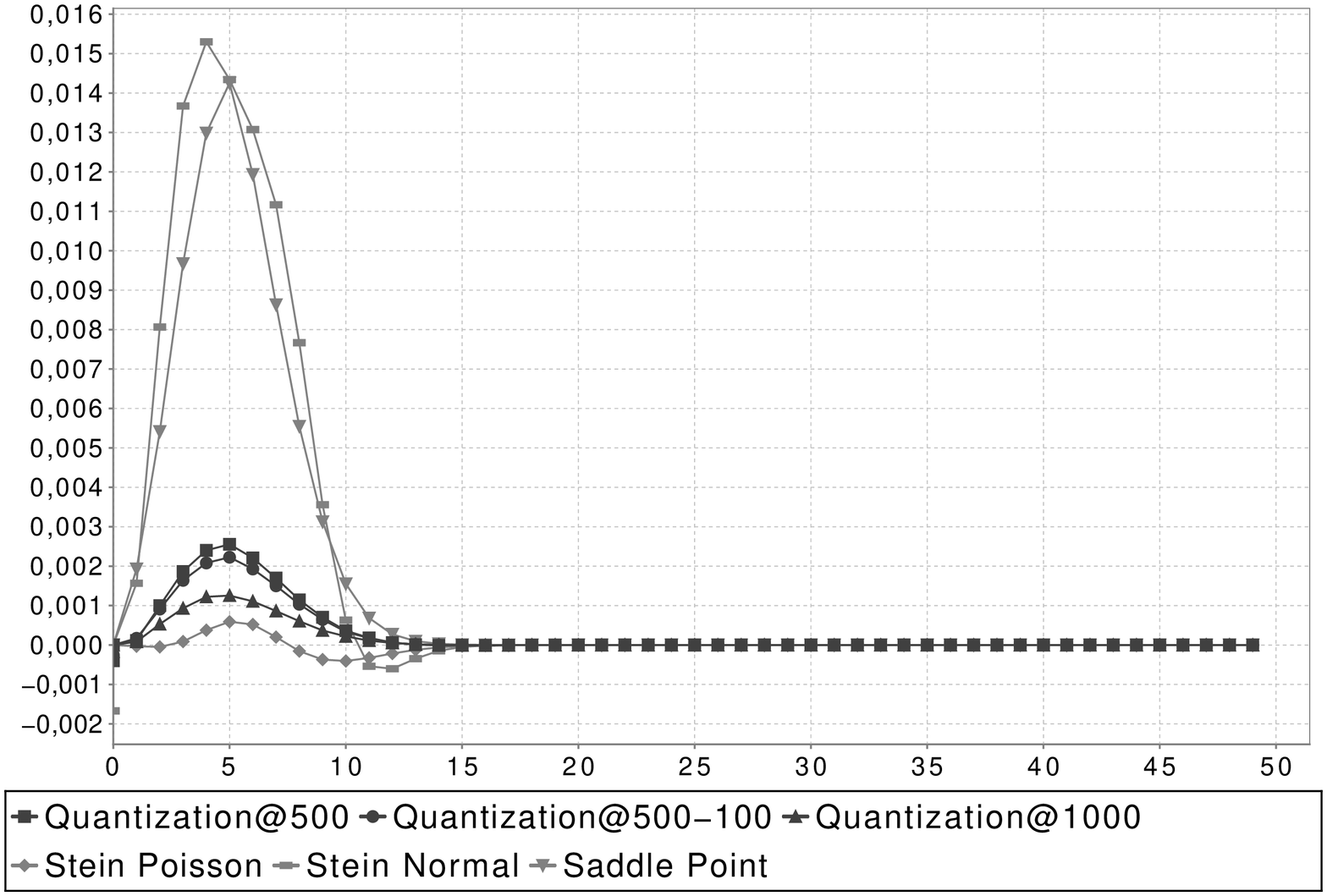}
    \caption{Absolute Errors for the call of various strikes ($p_0 = 0.05,
    \sigma = 0.5$).}\label{fig:homo1}
  \end{figure}

  \begin{figure}
    \centering
    \includegraphics[width=0.62\textwidth,angle=270]{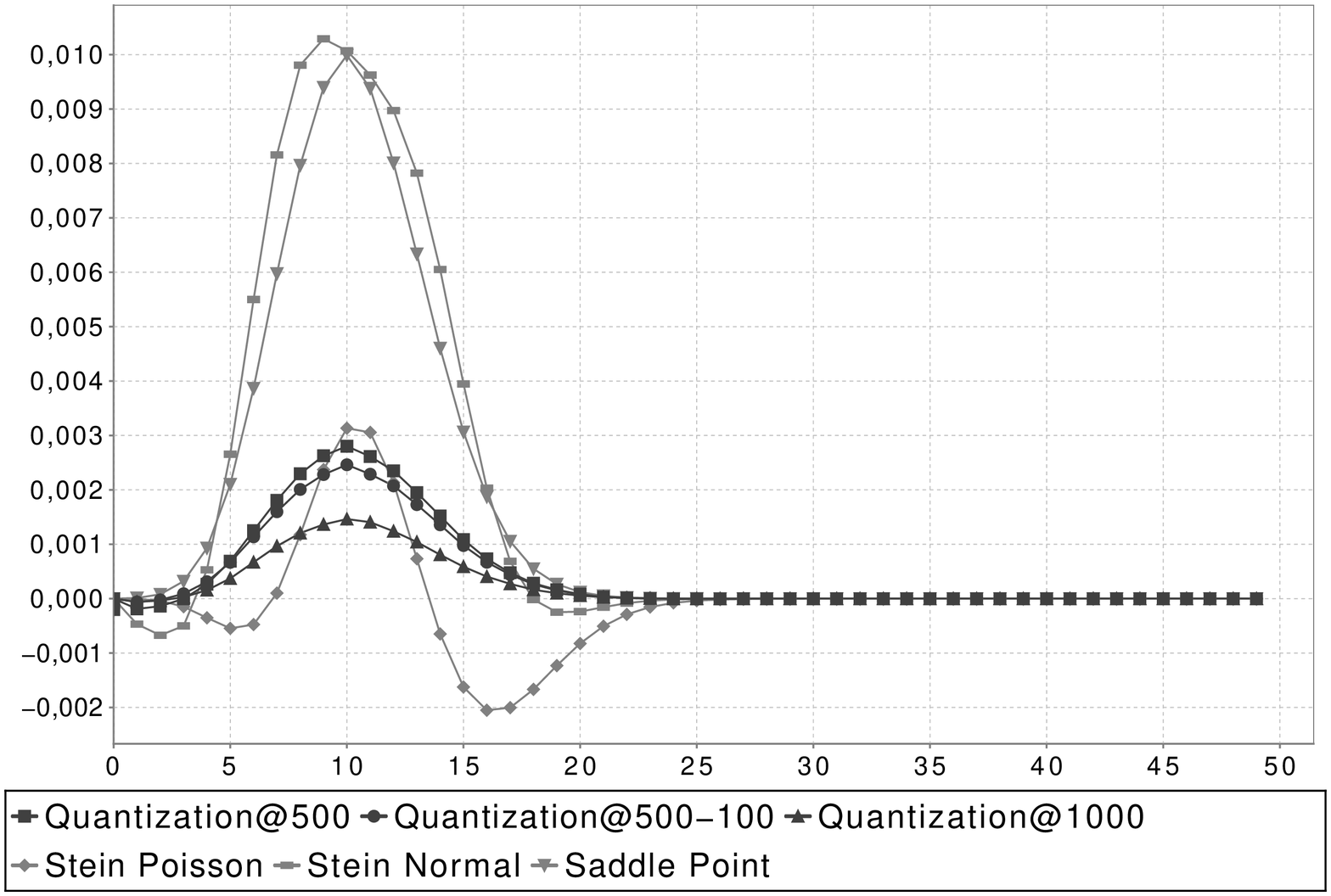}
    \caption{Absolute Errors for the call of various strikes ($p_0 = 0.10,
    \sigma = 0.5$).}\label{fig:homo2}
  \end{figure}

  \begin{figure}
    \centering
    \includegraphics[width=0.62\textwidth,angle=270]{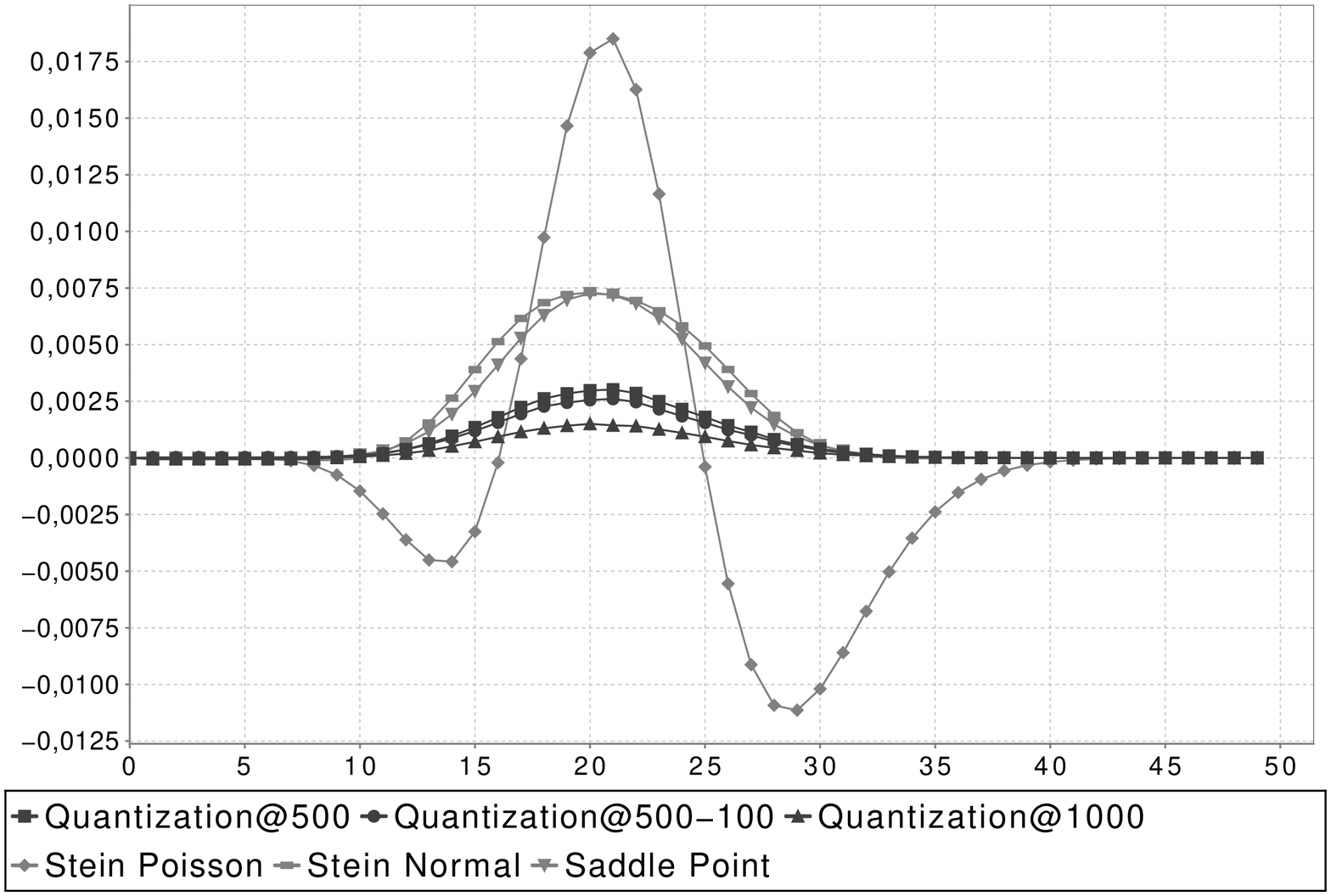}
    \caption{Absolute Errors for the call of various strikes ($p_0 = 0.20,
    \sigma = 0.5$).}\label{fig:homo3}
  \end{figure}

\subparagraph{Inhomogeneous random walk I.}
To discuss a more realistic scenario, we present an inhomogeneous setting with
$\alpha_i$ uniform distributed on the integers $\{1,2,\cdots, 10\}$, so that it
is still possible to compute some reference values by means of a recombining
binomial tree. The parameters read as follows

  \begin{itemize}
  \item $n= 100$,
  \item $\alpha_i \sim \mathcal{U}\{1,2,\cdots, 10\}$,
  \item $p_i = p_0 \, \exp(\sigma \xi_i - \sigma^2/2), \quad \xi_i \sim \Normal(0,1)$ i.i.d.,

    $p_0 \in \{0.05, 0.2 \}, \; \sigma = 0.5$,
  \end{itemize}

  The numerical results are depicted in Figures \ref{fig:inhomo1} and
  \ref{fig:inhomo2}.
   Note that we have excluded the Stein-Poisson
approach since this setting is already out of the
Poisson-limit domain for $p_0 = 0.05$ and consequently yield bad results.

 \begin{figure}
    \centering
    \includegraphics[width=0.62\textwidth,angle=270]{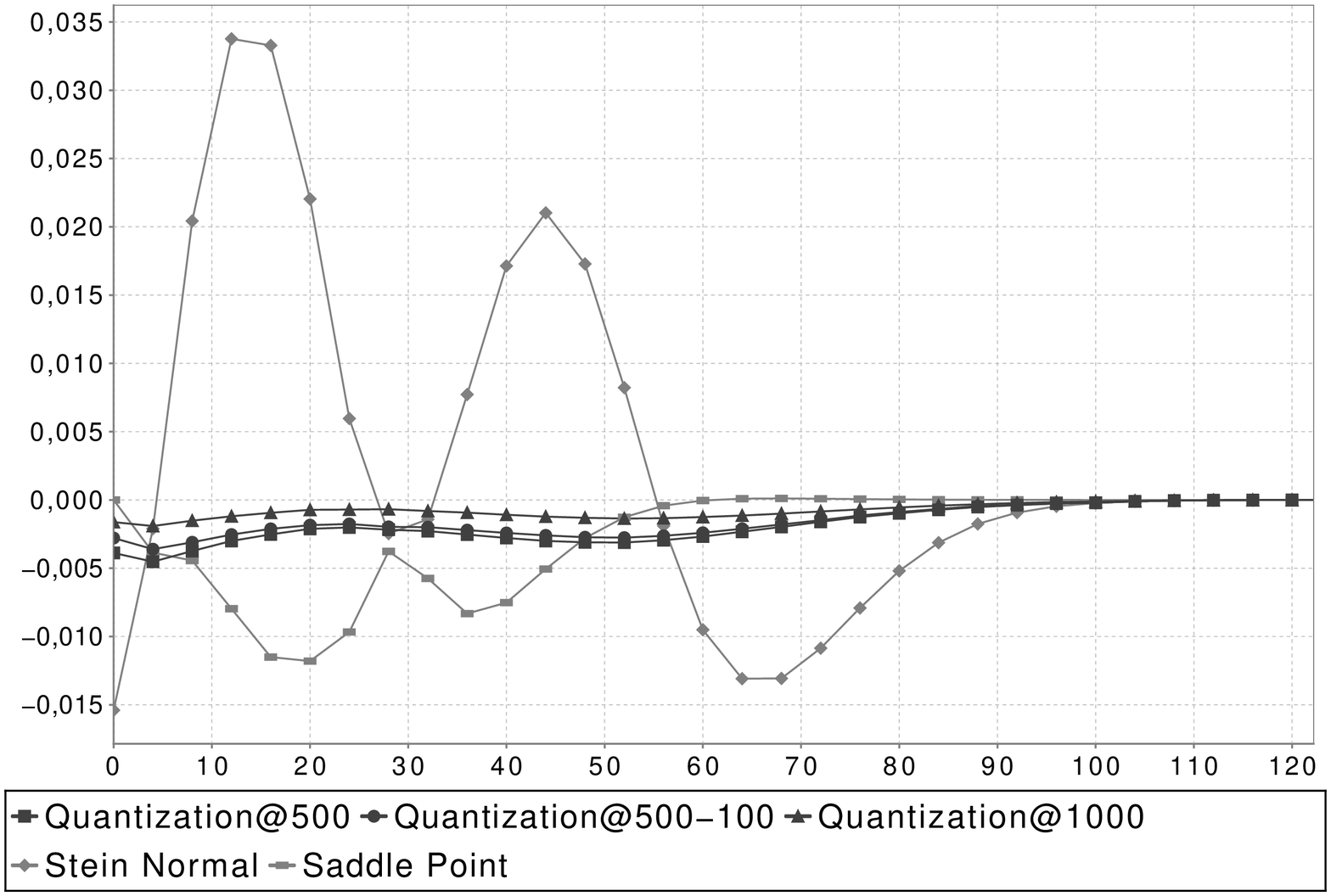}
    \caption{Absolute Errors for the call of various strikes ($p_0 = 0.05,
    \sigma = 0.5$).}\label{fig:inhomo1}
  \end{figure}

  \begin{figure}
    \centering
    \includegraphics[width=0.62\textwidth,angle=270]{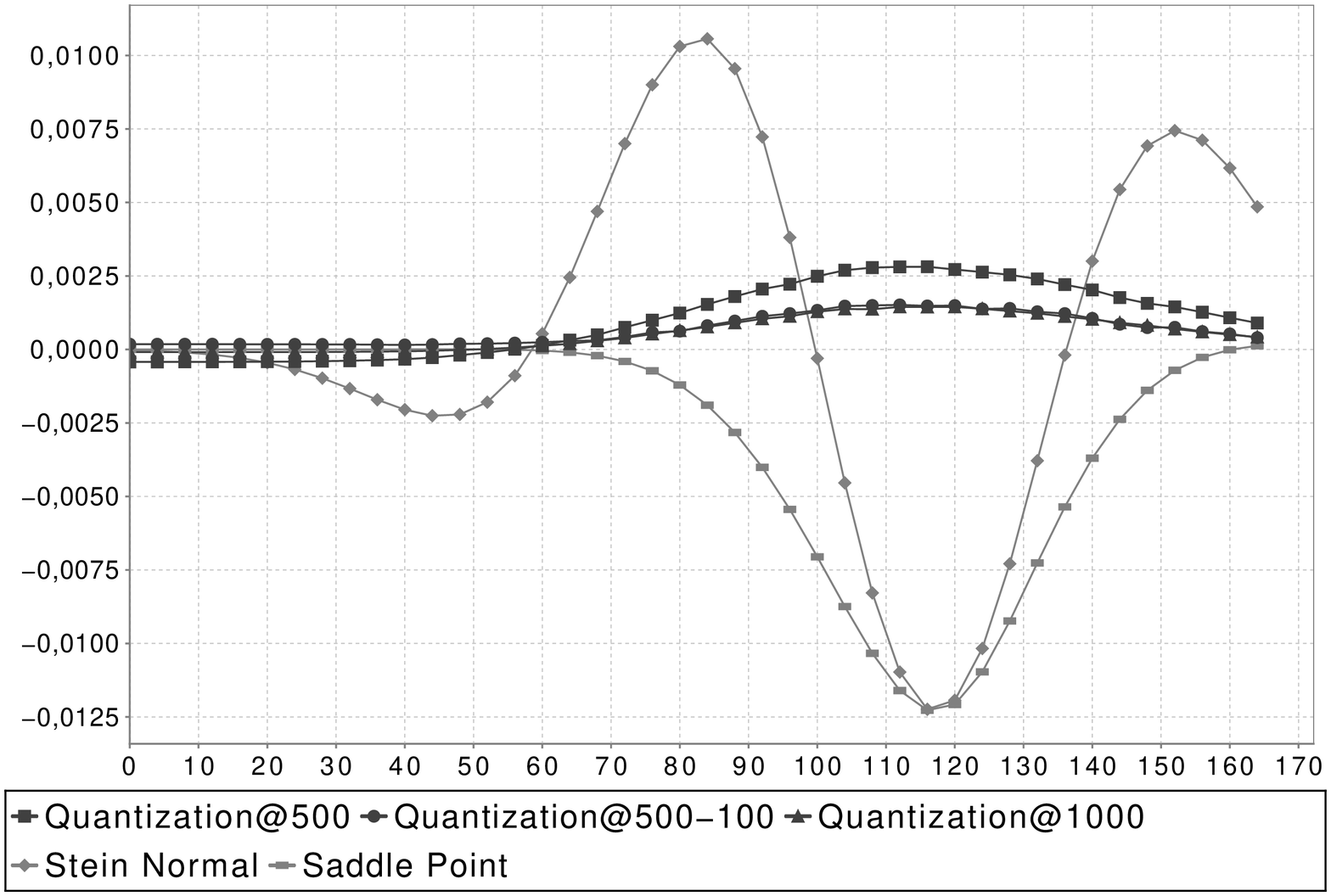}
    \caption{Absolute Errors for the call of various strikes ($p_0 = 0.20,
    \sigma = 0.5$).}\label{fig:inhomo2}
  \end{figure}

\subparagraph{Inhomogeneous random walk II.}
Finally, we present a non-trivial case, where the $\alpha_i$ are non-integer
valued any more, i.e. we have chosen them to be $\mathcal{U}([0,1])$
distributed. 
Since in this setting the recombining property of a binomial tree is destroyed,
we cannot compute the exact reference value any more.
Therefore, we have chosen a grid of size $N = 10000$ to compute
a reference value, since such a large grid size yields in all former
examples an absolute error less than $10^{-8}$.
To be more precise, the parameters has been chosen as follows:

 \begin{itemize}
  \item $n= 100$,
  \item $\alpha_i \sim \mathcal{U}([0,1])$,
  \item $p_i = p_0 \, \exp(\sigma \xi_i - \sigma^2/2), \quad \xi_i \sim \Normal(0,1)$ i.i.d.,

    $p_0 \in \{0.05, 0.2 \}, \; \sigma = 0.5$.
 \end{itemize}
 
 Since the Figures \ref{fig:qref1} and \ref{fig:qref2} are quite similar to
 those obtained in the first inhomogeneous setting (except a lower
 resolution), it seems very likely, that the former inhomogeneous setting
 is a very generic case to illustrate the general performance of the three
 tested methods.

 \begin{figure}
    \centering
    \includegraphics[width=0.62\textwidth,angle=270]{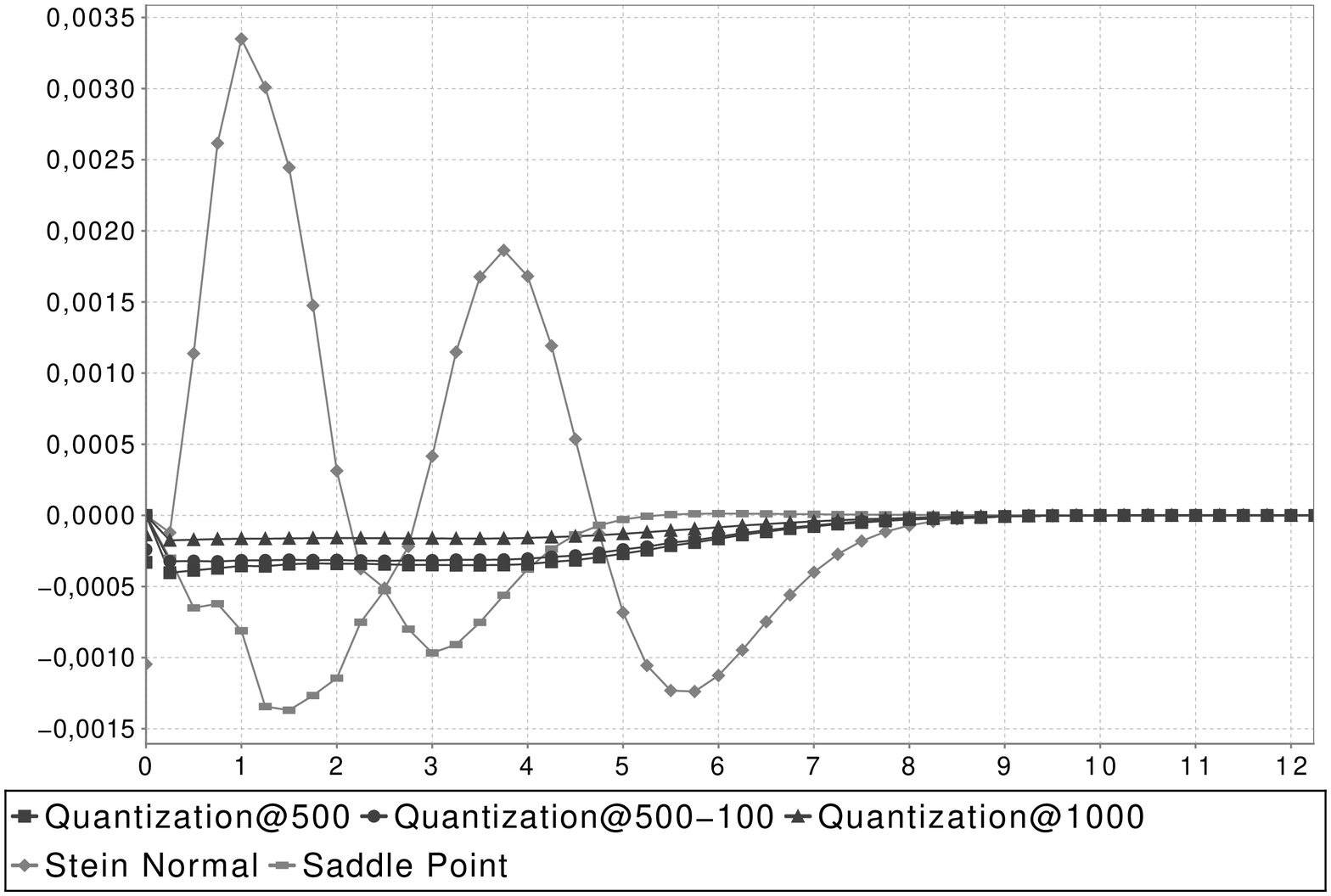}
    \caption{Absolute Errors for the call of various strikes ($p_0 = 0.05,
    \sigma = 0.5$).}\label{fig:qref1}
  \end{figure}

  \begin{figure}
    \centering
    \includegraphics[width=0.62\textwidth,angle=270]{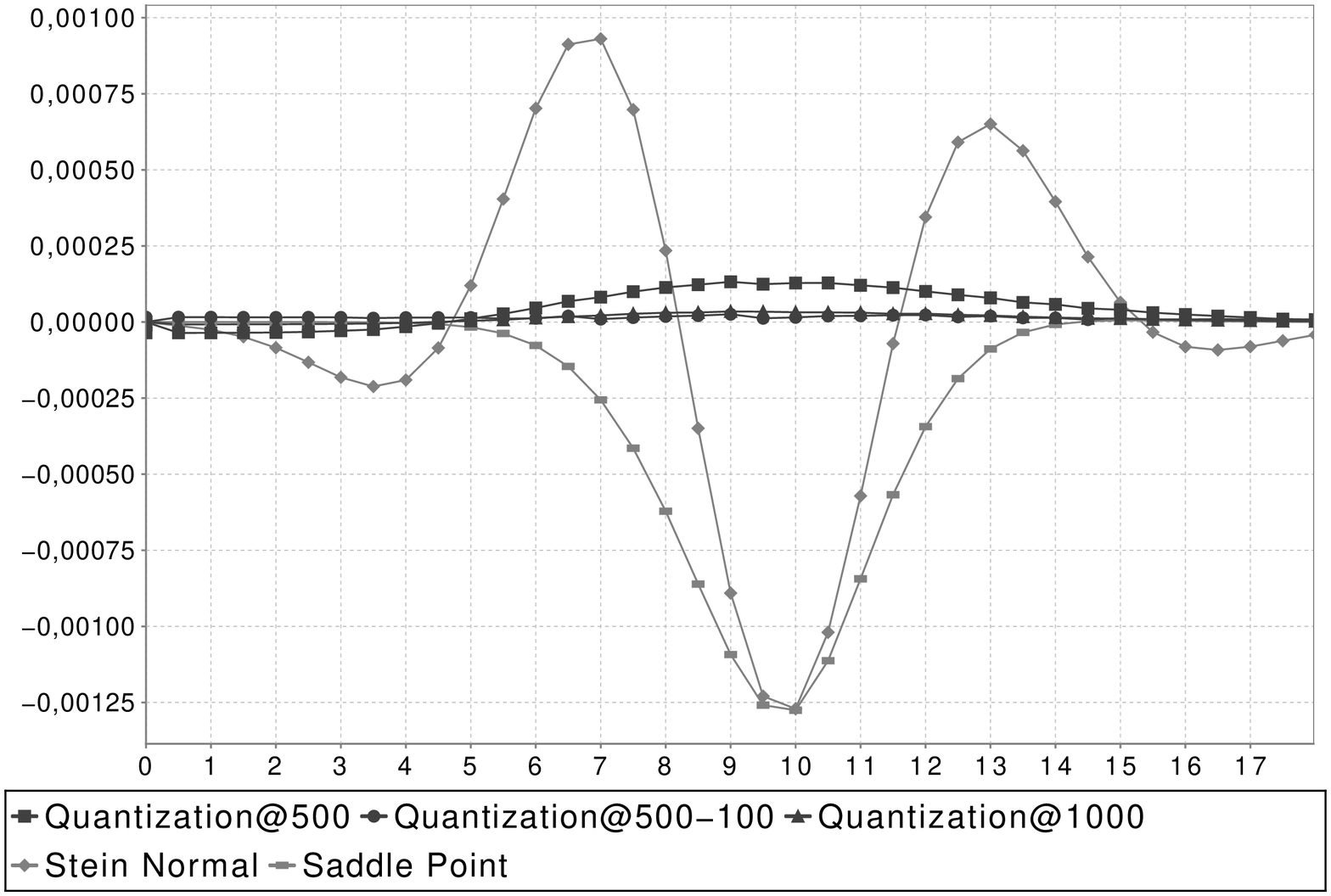}
    \caption{Absolute Errors for the call of various strikes ($p_0 = 0.20,
    \sigma = 0.5$).}\label{fig:qref2}
  \end{figure} 
 
 In all the above cases the quantization method remains very stable and
 outperforms even for a grid size of $N = 500$ in nearly all cases the
 other tested methods.
 Only in the homogeneous setting and for very small probabilities $p_i$, it
 cannot achieve the performance of the Stein-Poisson approximation.
 However, this excellence of the Stein-Poisson method in that particular
 setting is mainly caused by the fact, that the target distribution is an
 integer-valued one, as the Possion approximation is. Hence, these result are 
nontransferable to the inhomogeneous case.

In the more complex inhomogeneous setting (Figures \ref{fig:inhomo1} and
\ref{fig:inhomo2}), we still observe a strong domination of the quantization
methods for small and moderate probabilities. Furthermore, in the case $p_0 =
0.2$, we even get an error for the Romberg extrapolation with grid sizes $500$
and $100$, which is close to that of a $1000$-point quantization.

Concerning the computational time for the processing of our Dual Quantization
algorithm, this approach is of course not as fast as the Stein's method, where
one only needs to compute the two first moments of $X$ and then evaluates the
CDF-function of the standard normal distribution.
To apply our scheme, we have to process at each layer $k, 0\leq k \leq n$ at
full grid $\Gamma_k$ similar to recombining tree methods.
Nevertheless, the execution of Algorithm 1 implemented in {\tt C\#} on a {\tt
Intel Xeon} CPU@3GHz took for a grid size of $N=500$ only a few milliseconds.
Moreover, once the distribution of $\widehat X^n$ is established, we compute
$\E (\widehat X^n - K)_+$ for several strikes $K$ (as needed in practical
applications) in nearly no time.

Finally, we want to emphasize, that this approach gives, through the freedom to
choose a larger grid size, a control on the acceptable error for the
approximation.

\section{Approximation of the Greeks}

	Concerning the computation of sensitivities with respect to the parameters
	$\alpha_l$ and $p_l, 1\leq l\leq n$, we consider $f: \R^n_+ \times (0,1)^n \to
	\R,$ defined by
  \[
  (\alpha, p) \mapsto f(\alpha, p) := \E \Bigl( \sum_{i=1}^n \alpha_i Z_i -
  K\Bigr)^+.
  \]

  We are now interested in the computation of $\frac{\partial f}{\partial
  p_l}$ and $\frac{\partial f}{\partial \alpha_l} $.

Some elementary calculations reveal that for every $l\in \{0, \ldots, n\}$
  \[
  \frac{\partial f}{\partial p_l} = \E \Bigl( \sum_{i\neq l} \alpha_i Z_i -
  (K-\alpha_l)\Bigr)^+ -\E \Bigl( \sum_{i\neq l} \alpha_i Z_i - K\Bigr)^+
  \]
  and
  \[
  \frac{\partial f}{\partial \alpha_l} = p_l\cdot \Prob \Bigl( \sum_{i\neq l}
  \alpha_i Z_i \geq K-\alpha_l\Bigr),
  \]
so that our task consists of approximating the distribution of
\[
\partial_l X := \sum_{i\neq l} \alpha_i Z_i.
\]
This can be achieved
using a straightfoward adaption of the previous dual quantization tree, 
where we simply skip the $l$-th layer.

\com{{Quantization Scheme}
    Previous approach, where we skip the $k$-th layer}
   To be more precise, we set
\begin{equation*}
   \begin{split}
     \widehat{\partial_l X^0} & := 0\\
     \widehat{\partial_l X^k} & :=
     \begin{cases}
       \Ipk\!(\widehat X^{k-1} + \alpha_k Z_k)& k\neq l,\\
       \widehat{\partial_l X^{k-1}} & k = l.\\
     \end{cases}
  \end{split}
\end{equation*}

\begin{remark} 
This scheme can be processed simultaneously for all $l,\; 1\leq l \leq n$
without increasing the number of nearest neighbor searches.
\end{remark}

Numerical experiments, which are not reproduced here, also confirm the
good numerical performance of the dual quantization in this specific setting.

\subparagraph{Acknowledgements} We are very thankful to F.X. Vialard from
Zeliade Systems for helpful discussions and comments during our work on this
topic.

\bibliography{literatur}

\end{document}